\documentclass[conference]{IEEEtran}
\usepackage{graphicx}
\usepackage{amssymb}
\usepackage{tikz}
\usepackage{bbold}
\usepackage{amsthm}
\usepackage{amsmath}
\usetikzlibrary{arrows}
\newtheorem{theorem}{Theorem}

\newtheorem{cor}{Corollary}

\newtheorem{conjecture}{Conjecture}
\usepackage{chngcntr}
\counterwithin{paragraph}{subsubsection}
\counterwithin{subsubsection}{subsection}
\tikzstyle{block} = [draw, fill=gray!10, rectangle, minimum height=1cm, minimum width=1.5cm]
\usepackage{comment}

\usepackage[
top    = 0.74in,
bottom = 1.09in,
left   = 0.602in,
right  = 0.602in]{geometry}

\begin{document}
\def\whole{}

% paper title
\title{Source Broadcasting to the Masses: Separation has a Bounded Loss}

\author{\authorblockN{Uri Mendlovic and Meir Feder}
\authorblockA{Department of Electrical Engineering - Systems\\
Tel-Aviv University, Tel-Aviv 69978, Israel \\
Email: \{urimend,meir\}@eng.tau.ac.il}}
\maketitle

\begin{abstract}

This work discusses the source broadcasting problem, i.e. transmitting a source to many receivers via a broadcast channel. The optimal rate-distortion region for this problem is unknown. The separation approach divides the problem into two complementary problems: source successive refinement and broadcast channel transmission.
We provide bounds on the loss incorporated by applying time-sharing and separation in source broadcasting. If the broadcast channel is degraded, it turns out that separation-based time-sharing achieves at least a factor of the joint source-channel optimal rate, and this factor has a positive limit even if the number of receivers increases to infinity.
For the AWGN broadcast channel a better bound is introduced, implying that all achievable joint source-channel schemes have a rate within one bit of the separation-based achievable rate region for two receivers, or within $\log_2 T$ bits for $T$ receivers.
%This result is in line with previous works that showed simple single letter achievability bounds within $1$ bit for various AWGN %network information problems.

%The {\bf Abstract} section should be no more than $250$ words and

\end{abstract}

\section{Introduction}

\ifdefined\whole
\subsection{Background}
We briefly review the source successive refinement and the degraded broadcast coding problems which are used as the separation approach for source broadcasting.
In successive refinements of a source, let $X$ be a source with known distribution and let $d(x, \hat{x})$ be a distortion measure. One would like to construct $T$ messages or information streams from $X$. These messages are assumed to arrive successively from $T$ to $1$, i.e, if a certain message $W_t$ is available, all messages $W_{t'}$ with $t' > t$ are available as well. For given distortion values $D_1, \ldots, D_t$ (denoted $D_1^T$), where $D_t$ is the maximal allowed distortion resulted from decoding messages $W_t^T$, we ask what the minimal message information rates $R_1^T$ are, where $R_t$ is the information rate of $W_t$. According to \cite{SR91}, some sources (with a certain distortion measure) are successively refinable, meaning that any distortion value is accepted at the same total rate of the classical rate-distortion function, as if successive refinement was not carried. An equivalent definition of successive refinability is the statistical independence of the constructed messages.

The complementary channel coding problem is the degraded broadcast channel coding. In broadcast channel we assume one transmitter and $T$ receivers. The encoder receives $T$ independent messages $W_1^T$ and encodes them into the transmitted signal $X$. Each channel output $Y_t$ is available at receiver $t$, and the $t$-th decoder should recover message $W_t$. The channel is defined according to $T$ channel distributions $\{w_t(y_t | x)\}_{t=1}^T$. We call a broadcast channel stochastically degraded if one can choose a matching joint distribution $w(y_1, ..., y_T | x)$ that constructs a Markov chain $X-Y_1-...-Y_T$ consistent with the marginal distributions $\{w_t(y_t | x)\}_{t=1}^T$. The capacity region of degraded broadcast channels is known \cite{Cover72} \cite{Cover98}.
\subsection{Contribution}
\fi

This work studies the problem of distortedly transmitting a source with known distribution and distortion measure to $T$ receivers, through a broadcast channel.
The optimal solution for this problem, attained by a joint source-channel coding, is unknown. This paper examines the loss attained by using the sub-optimal separation approach for this problem which divides the problem into source successive refinement (which is a source coding problem) and broadcast channel coding.
A special interest is devoted to the case where the number of receivers is large, showing that the separation loss is bounded even as the number of receivers grows to infinity.

We assume in the sequel that each receiver has a distortion requirement $D_t$ it cannot exceed. We use the traditional notation for the source rate-distortion function $R(D)$, as defined in the point-to-point sense, and the shortened notation $R_t = R(D_t)$ to denote the rate distortion associated with the $t$-th receiver.
The meaningful quantity in source-channel coding is the \emph{rate}, denoted $r$, of transmitted source symbols per channel use. This quantity measures the performance of a coding scheme, because it measures the amount of channel resource needed to achieve the required distortions.

The first part of the paper studies the case of combining a successively refinable source with a general degraded broadcast channel.
%In this case we define the first receiver $t = 1$ as the best receiver, and all other receivers are sorted in decreasing order. Thus, meaningful problem setup requires $D_1 \leq D_2 \leq \ldots \leq D_T$, and thus $R(D_1) \geq R(D_2) \geq \ldots \geq R(D_T)$.
%We consider a separation approach that divides the problem into source successive refinement (source coding) and broadcast channel coding (channel coding).
%Note that when using a separated source coding and channel coding, we use $W_t$ to refer to the message designated to be transmitted to the $t$ receiver. Due to the degraded nature of the channel, $W_t$ will be decodable at all of the receivers before $t$.
In this part we provide bounds on the loss incorporated by applying separation and even time-sharing compared with the optimal (unknown) joint coding for source broadcasting. Interestingly, it turns out that in this case separation-based time-sharing achieves at least a (small) multiplicative factor of the joint source-channel optimal rate, and this factor has a positive limit even if the number of receivers increases to infinity.

The second part of the paper considers broadcasting a successively refinable source over an additive white Gaussian noise (AWGN) broadcast channel. The achievable rate and the optimal coding method are unknown in this case as well. For this AWGN case we prove a stronger additive loss bound stating that any achievable rate, measured per complex degree of freedom, is within one bit of the separation-based achievable region for two receivers, or within $\log_2 T$ bits for $T$ receivers.

The conclusion of these results is that separation-based codes and even applying time-sharing for source broadcasting to many receivers provide a good practical compromise between simple design and computational complexity and transmission performance, and this compromise holds regardless of the amount of receivers.

\subsection*{Existing Work}

Many works studied the problem of joint source-channel coding for source broadcasting, some regards the case of successive refinement through broadcasting. Some works propose practical coding techniques and schemes for a concrete source or channel \cite{HDA}, \cite{expansion}.

Other works considered this problem from an information theoretic perspective, providing performance bounds and partial capacity region characterization. Tian \textit{et al.} studied the performance of separation-based coding scheme for the special case of Gaussian sources \cite{Tian09} and more generally for sources with mean squared error distortion measure \cite{Tian10}. Together with a general bound, we present a particular result for the AWGN broadcast channel which is analogous to their work, but the bound in \cite{Tian09} \cite{Tian10} holds for a special type of source and arbitrary channel, while our result is given in terms of the Gaussian channel, and for arbitrary successively refinable source.

\subsection*{Paper outline}

The rest of this paper is organized as follows. Section \ref{sec:bounds} provides multiplicative bounds on the rate of separation/time-sharing compared with the joint source-channel coding rate.
\ifdefined\whole
Section \ref{sec:discussion} explores some properties of the bounds.
\fi
Section \ref{sec:awgn} addresses the AWGN broadcast channel by introducing a tighter additive bound. Finally, Section \ref{sec:future} concludes the paper and suggests some future work.

\section{Multiplicative Rate Bounds}
\label{sec:bounds}

\subsection{A Naive Lower Bound}
\label{subsec:naive_bound}

We suggest the following trivial lower bound on the rate achievable using separation, which holds for both degraded and non-degraded broadcast channels:

\begin{theorem}
\label{theorem:naive_bound}
Source-channel separation in broadcast channel reduces the optimal achievable rate by a factor not greater than the number of receivers.
\end{theorem}
\begin{proof}
One can apply point-to-point channel separation on any single receiver, achieving the original source symbols to channel symbols rate while transmitting only to that single receiver at its desired distortion.
By applying this method on every receiver, the rate is reduced by a factor equals to the number of receivers.
\end{proof}

Note that Theorem \ref{theorem:naive_bound} holds for the non-degraded broadcast scenario as well as for the degraded scenario. It holds for arbitrary correlated messages with arbitrary distortion measures, including the case of a non-successively refinable source.

The only requirement for Theorem \ref{theorem:naive_bound} to hold is that the original problem can be described as a group of individual problems, one for each receiver, and each of these point-to-point problems has the separation property. In our case of source refinement this holds because we split the problem into point-to-point rate-distortion problems.

\subsection{A Refined Lower Bound}
\label{subsec:bc:refined_bound}

We would like to refine Theorem \ref{theorem:naive_bound} for the case of degraded broadcast channels with a successively refinable source.

First assume only two destinations, $T=2$. Consider a source that is transmitted at a certain source symbols per channel symbols rate $r$ without applying separation. Using the standard successive refinements method, one can encode the source into 2 messages $W_1, W_2$, and then transmit $W_2$ to the second receiver with a point-to-point channel code. Applying the point-to-point source-channel separation theorem on the second (degraded) receiver, this scheme can accomplish the second receiver requirement at rate of no less than the original rate $r$.

Due to the degraded nature of the channel, $W_2$ is decodable at the first receiver as well, thus the first decoder can achieve its desired distortion by receiving $W_1$. Since the joint source channel scheme enables a distortion $D_1$ at the first receiver with $r$ symbols of per channel use, the channel to the first receiver can carry an information rate of $R_1 \cdot r$ bits per channel use, where, as we recall, $R_1, R_2$ are the rate distortions of the source for the decoders distortion requirements, respectively.
%
%according to the point-to-point channel coding theorem, the transmitter can send $W_1$ to the first receiver at a rate %of $r$ symbols of per channel use, achieving an information rate of $R_1 \cdot r$ bits per channel use, where $R_1, %R_2$ are the rate distortion of the source for the decoders distortion requirements, respectively.
%
Note that because the source is successively refinable, $W_1, W_2$ have information rates of $(R_1-R_2), R_2$, respectively. Recalling that the first receiver already received $W_2$ when it had been transmitted to the second receiver, it only needs another $\frac{R_1 - R_2}{R_1 \cdot r}$ channel uses per source symbol to transmit the missing information. The total amount of channel uses per source symbol in this separation scheme is
\[\frac{1}{r} + \frac{R_1 - R_2}{R_1 \cdot r}, \]
and thus the original transmission rate $r$ is reduced by a factor of
\[1 + \frac{R_1 - R_2}{R_1} \]

This bound coincides with the bound from Theorem \ref{theorem:naive_bound} when the successive refinement messages satisfy $R_1 \gg R_2$. But as the desired distortions become similar to each other, i.e. when their rate-distortion values are close,  $R_1 \approx R_2$, separation loss becomes negligible according to the refined bound. In the Theorem below this result is generalized to the following lower bound:

\begin{theorem}
\label{theorem:refined_bound}
In a $T$-users successive refinement over degraded broadcast channel, if the rate-distortion of the $t$-th user equals to $R_t$ then applying source-channel separation reduces the achievable rate by a factor not greater than:
\begin{equation}
\label{eq:refined_bound_RT}
1 + \sum_{t=1}^{T-1} \frac{R_t-R_{t+1}}{R_t} =
T - \sum_{t=1}^{T-1} \frac{R_{t+1}}{R_t}
\end{equation}

If the receiver's rate distortion values are in an interval $[R_{min}, R_{max}]$, then  this factor can be bounded by
\begin{equation}
\label{eq:refined_bound_T}
T - (T-1) \cdot (R_{min}/R_{max})^{1/(T-1)},
\end{equation}
which approaches
\begin{equation}
\label{eq:refined_bound}
1 + \ln(R_{max}/R_{min})
\end{equation}
as the number of receivers increases to infinity.
\end{theorem}

\begin{proof}
We broaden the above argument to arbitrary number of users $T$. First transmit the complete message to the last, most degraded receiver, using the original number of channel uses per source symbol $1/r$. Then, transmit the remaining $R_{T-1}-R_T$ information to the $(T-1)$-th receiver. Because the original unseparated scheme was able to transmit to this receiver at information rate of $R_{T-1} \cdot r$ per channel use, we can carry the additional information at $(R_{T-1}-R_T) / (r \cdot R_{T-1})$ channel uses per source symbol. In this manner the $t$-th receiver can accomplish its distortion limit after the $(t-1)$-th receiver is satisfied by adding at most $(R_t-R_{t+1}) / (r \cdot R_t)$ channel uses per source symbol. Summing up all the transmission phases, the total number of channel uses per source symbol is:
\begin{equation*}
\label{eq:ratio_series}
\frac{1}{r} + \sum_{t=1}^{T-1} \frac{R_t-R_{t+1}}{r \cdot R_t},
\end{equation*}
leading to the upper bound (\ref{eq:refined_bound_RT}) on the rate reduction ratio.

To maximize this bound one should choose $R_1=R_{max}$ and $R_T=R_{min}$. As for the maximization over remaining $R_2,\ldots,R_{T-1}$, for each $t$, the expression
\[ - \frac{R_{t+1}}{R_t} - \frac{R_t}{R_{t-1}}\]
is maximized for
\[ R_t = \sqrt{R_{t-1}R_{t+1}} \]
so the bound is maximized by choosing a geometric sequence of rate distortions:
\[ \frac{R_{t+1}}{R_t} = (R_{min}/R_{max})^{1/(T-1)} \quad\forall t=1,..,T-1\]
Substituting this ratio in (\ref{eq:refined_bound_RT}) results in (\ref{eq:refined_bound_T}). Eq. (\ref{eq:refined_bound}) is immediately derived by recalling that
\begin{multline*}
\lim_{T \to \infty} (T-1) \cdot (1 - (R_{min}/R_{max})^{1/(T-1)}) = \\
\ln(R_{max}/R_{min})
\end{multline*}
\end{proof}

\subsection{A Channel-Oriented Bound}
\label{subsec:channel}

In the above theorems we analyzed the performance of a simple but useful separation-based approach for broadcasting successive refinements of a source. We bounded the loss relatively to the optimal performance in terms of source rate-distortion values of interest. An analogous argument can be stated in terms of the channel \emph{point-to-point capacities} $\{C_t\}$, i.e. the capacities of transmission to each receiver. Using the time-sharing method described above, it is easy to see that all above results are valid, where $R_t, R_{max}$ and $R_{min}$ are replaced by their respective terms of the channel point-to-point capacities $C_t, C_{max}$ and $C_{min}$.

Given a source-channel separation scenario, one can compute two bounds: using the rate-distortion values $\{R_i\}$ and using the point-to-point capacities $\{C_i\}$. In fact, a more precise bound can be stated using both terms:

\begin{theorem}
\label{theorem:capacity_bound}
In a $T$-users successive refinement over degraded broadcast channel, if the rate-distortion of the $t$-th user equals $R_t$ and its point-to-point capacity is $C_t$ then applying source-channel separation reduces the achievable rate by a factor not greater than:
\begin{equation}
\label{eq:capacity_bound_RT}
\frac
{\frac{R_T}{C_T} + \sum_{t=1}^{T-1} \frac{R_t-R_{t+1}}{C_t}}
{\max_{t=1,\ldots,T} \left( \frac{R_t}{C_t} \right)}
\end{equation}
\end{theorem}

\ifdefined\whole
\begin{proof}
By applying the point-to-point source channel separation on the $t$-th receiver, a joint source-channel coding scheme can transmit no more than $R_t/C_t$ source symbols per channel use. Combining all these limitations, joint coding cannot exceed
\[ \max_{t=1,\ldots,T} \left( \frac{R_t}{C_t} \right) \]
channel uses per source symbol.

On the other hand, one can use the presented separation method: starting by transmission to the most remote receiver, which demands $R_T/C_T$ channel uses per source symbol. Inductively, the $t$-th receiver demands additional information rate of $R_t-R_{t+1}$ bits beyond the requirement of the $(t+1)$-th receiver. This information rate can be transmitted to the $t$-th receiver using $(R_t-R_{t+1}) / C_t$ channel uses per source symbol. Thus this separation method achieves a rate of
\[ \frac{R_T}{C_T} + \sum_{t=1}^{T-1} \frac{R_t-R_{t+1}}{C_t} \]
channel uses per source symbol.

Combining the best performance of the joint coding with the achievable result of the presented separation method, the ratio of the achievable rates of joint coding versus separation-based coding is no more than the value given by eq. (\ref{eq:capacity_bound_RT}).
\end{proof}
\else
The proof of this Theorem is given in the complete version of this paper, \cite{Complete}.
\fi

For a given rate-range $[R_{min}, R_{max}]$ and point-to-point capacities range $[C_{min}, C_{max}]$, it can be easily shown that the worst-case maximal bound of Theorem \ref{theorem:capacity_bound} coincides with the result of Theorem \ref{theorem:refined_bound} (i.e., (\ref{eq:refined_bound_T}) and (\ref{eq:refined_bound})) with the bigger ratio:
 \[T - (T-1) \cdot \max ( R_{max}/R_{min}, C_{max}/C_{min})^{1/(T-1)}.\]

\ifdefined\whole
\section{Discussion}
\label{sec:discussion}

\subsection{Tightness}
\label{subsec:tightness}
\else
\subsection{Discussion}
\label{subsec:discussion}
\fi
Fig. \ref{fig:classes} demonstrates the hierarchy of coding classes. Joint source channel coding is the most general class, and thus achieves the maximal performance. One can separate the coding into source-coding and channel-coding. This class has a single letter capacity expression, and coding schemes in this class may be simpler to implement. The third most-limited class of interest is separated coding schemes that use a time-sharing approach for the channel coding. The results of section \ref{sec:bounds} actually bound the ratio between optimal performance of the first and third classes, thus incidentally bound the ratio of the first and second classes, namely the separation loss.

\ifdefined\whole
We argue that any separation-based coding may outperform the result stated in theorem \ref{theorem:refined_bound} only due to improvement upon abandonment of the time-sharing nature of our solution. The optimal separation strategy will have actual performance according to the minimal amount of channel uses per symbol required to carry information at the rate-tuple $(R_1-R_2, ..., R_{T-1}-R_T, R_T)$ using the well-known degraded broadcast coding scheme.
\fi

\begin{figure}
\centering
\includegraphics[width=0.8\linewidth]{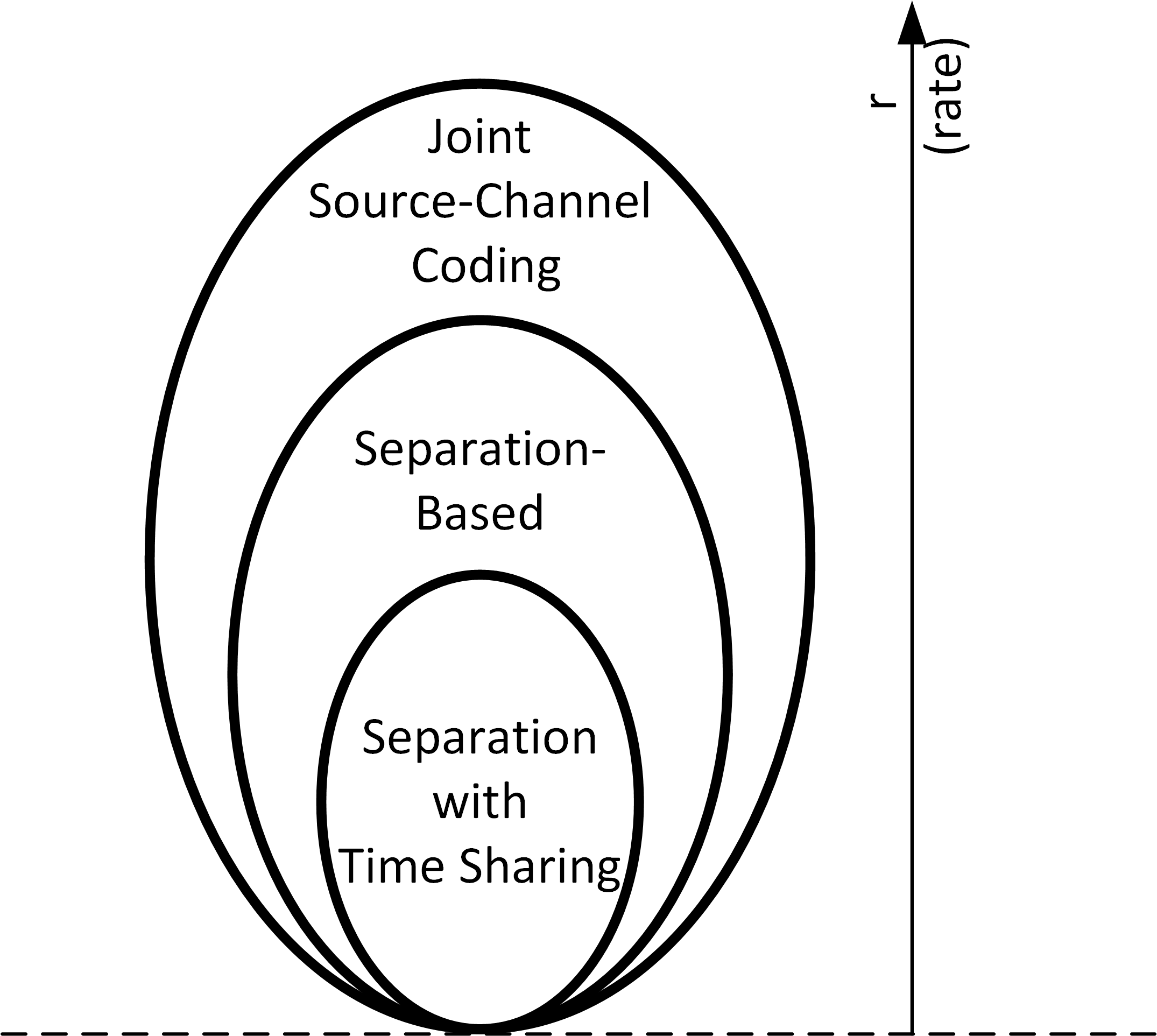}
\caption{Euler diagram of coding classes.}
\label{fig:classes}
\end{figure}

In some broadcast channels time-sharing is optimal or near-optimal. For example, in erasure broadcast channel time-sharing is optimal \cite{BEC}, and in spread-spectrum broadcasting it is near-optimal \cite{Lapidoth}. In these cases our bound becomes tight because the degraded broadcast capacity region is a simplex.

When the actual capacity region is significantly better than time-sharing, separation may cause much less loss than the suggested bounds. In the extreme case, the capacity region is a separable region, and separation will inflict no loss, while our bound may show a significant loss. In section \ref{sec:awgn} we discuss explicitly the Gaussian example. It turns out that for low SNR our bound is tight with the loss of separation, while at high SNR separation approaches the performance of the optimal joint source-channel scheme.

\subsection{Non-Degraded Broadcast Channel}
\label{subsec:non-degraded}
Although Theorem \ref{theorem:naive_bound} is applicable to any source-refinement-via-broadcasting scenario, Theorems \ref{theorem:refined_bound} and \ref{theorem:capacity_bound} require a successively refinable source and a degraded broadcast channel. In this subsection we show that Theorem \ref{theorem:refined_bound} also holds for non-degraded broadcast channels.

The problem of transmitting a successively refinable source through a non-degraded broadcast channel discussed earlier in \cite{KM77}, as a special case of transmitting \emph{degraded messages} over non-degraded broadcast channel. Therefore applying separation on this scenario breaks it into two solved problems: successive refinement and transmitting degraded messages over broadcast channel.

A careful inspection of Theorem \ref{theorem:refined_bound} reveals that the inductive argument holds in this case as well.
Suppose that $W_{t+1},\ldots,W_T$ has already been transmitted to the first $t$ receivers and we would like to transmit $W_t$ to the first $t$ receivers. For degraded channels we argued that this can be done using no more than $(R_t-R_{t+1}) / (r \cdot R_t)$ channel uses per source symbol. This bound holds for non-degraded channels as well, because this problem may be seen as transmitting $W_t$ through a \emph{compound channel} combining the channels of the first $t$ receivers. If the optimal joint source-channel coding operates at the rate of $r$ source symbols for channel use, it is at least able to transmit $W_t^T$ to the first $t$ receivers at this rate, achieving information rate of $R_t$ bits per channel use in this compound channel. Because source-channel separation holds for compound channels \cite{compound}, a separation-based scheme can achieve this information rate for transmitting $W_t$. Because the information of $W_t$ is $(R_t-R_{t+1})$ one needs no more than $(R_t-R_{t+1}) / (r \cdot R_t)$ channel uses per source symbol of $W_t$.

\ifdefined\whole
\subsection{Span of Rates}
\label{subsec:scalability}
Note that our bound may suggest that separation is a wasteful approach in the general case, but in practical systems the transmitted data has a characteristic rate-distortion region that allows only a limited loss. Even as the rate scale $R_{max}/R_{min}$ becomes larger the separation loss bound increases slowly due to the logarithmic behavior of the bound.

For example, compressed video data rates may vary between dozens of Kbits per second for the minimal quality videophone to dozens of Mbits per second for the finest high definition (HD) video. Although the rate distortion function has three orders of magnitude scale ($R_{max}/R_{min} = 1000$), our bound (\ref{eq:refined_bound}) results in a factor of $8$, regardless of the number of receivers. Looking at the bound from the channel capacities perspective, a TV signal may be broadcasted to a region where, even in a spread non-homogenous region the SNR may vary by, say $30-40$ dB. Depending on the SNR of the worst receiver, our bound indicates a rate loss by a factor of about $3$ in these extreme conditions. In the next section we introduce a tighter bound for the AWGN scenario.
\fi

\section{Additive Rate Bound for the AWGN Broadcast Channel}
\label{sec:awgn}
\ifdefined\whole
\subsection{Distance Bound}
\label{subsec:awgn_bound}
\fi
 Consider now a degraded additive white Gaussian noise broadcast channel with input power $P$ and noise power $N_t$ at the $t$-th receiver.
%In the general case of transmitting a (not necessary successively refinable) source through this channel the achievable rate and coding method are unknown. 
In this section we prove that any achievable rate measured per complex degree of freedom is within one bit of the separation-based achievable region for two receivers, or within $\log_2 T$ bits for $T$ receivers.

\begin{figure}
\centering
\includegraphics[width=0.7\linewidth]{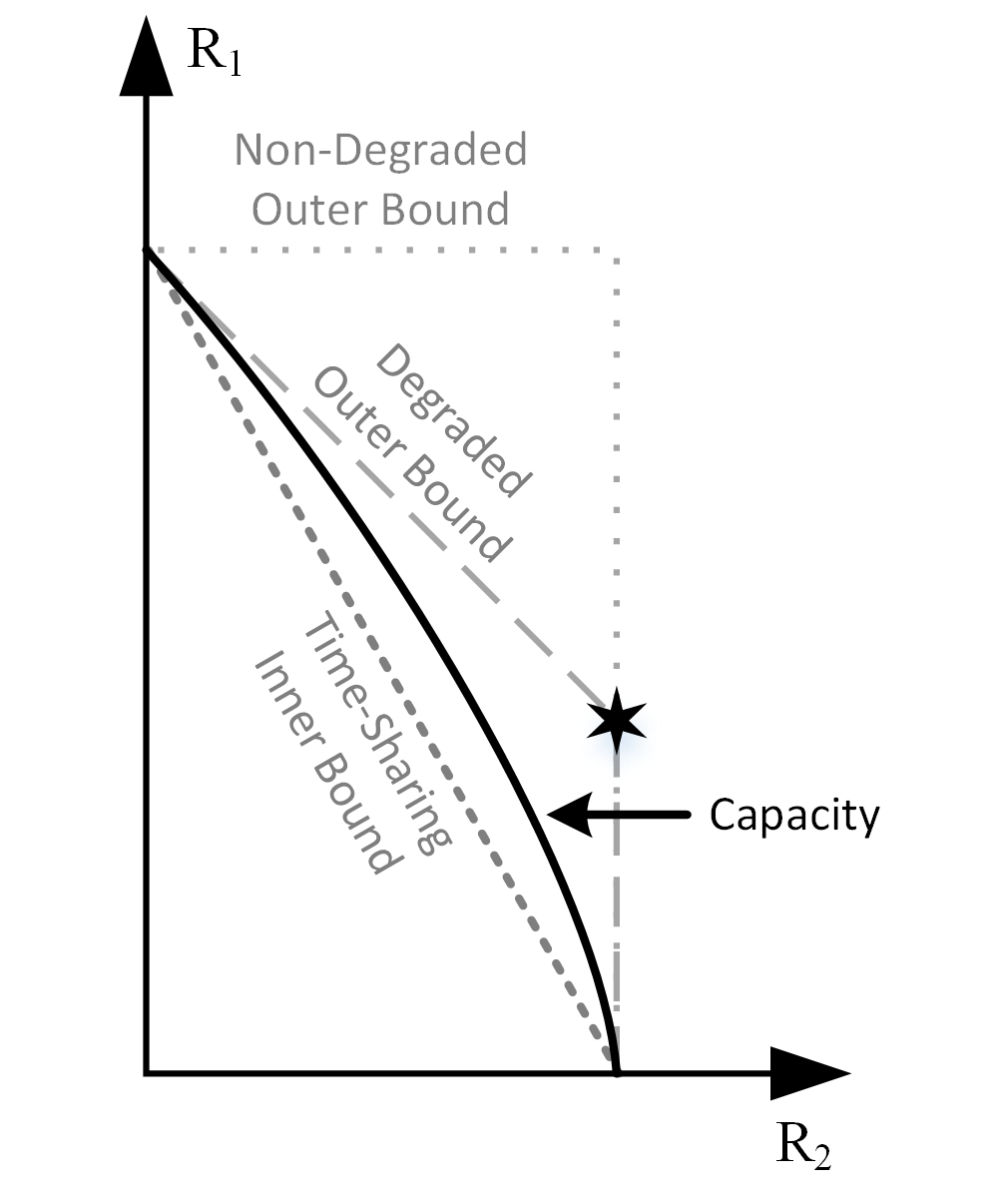}
\caption{The studied regions of the rate-space and the best achievable point for joint coding over a degraded channel.}
\label{fig:regions}
\end{figure}

To bound the difference we study the best achievable point, defined as a rate tuple that simultaneously achieves all point-to-point capacities. This point is marked with a star in Fig. \ref{fig:regions}. Theorem \ref{theorem:refined_bound} bounds the distance of this point from the time-sharing achievable region (relative to the distance of the point from the origin). In this section we bound the absolute distance of the point from the separation-based capacity region.

\begin{theorem}
\label{theorem:awgn}
For any transmission of a source to $T$-receivers over an degraded AWGN broadcast channel, if a distortion-tuple $D_1,\ldots,D_T$ is simultaneously achievable using $1/r$ channel uses per source symbols, then the derived rate-tuple, defined as the rate-distortion of each receiver normalized to a single channel use ($rR(D_t)$) is within
\[ C_{AWGN}\left(\frac{P}{P + N_T}
\sum_{t = 1}^{T-1} \left( 1 - \frac{N_t}{N_{t+1}}\right) \right) \]
of the classical separation-based capacity region of the channel, where $C_{AWGN}(x)=\log(1+x)$.
\end{theorem}

\begin{proof}
In order to bound the distance of an achievable point from the classical capacity region, we compute the shift of the point from the capacity region along the direction of the most degraded $T$-th user axis.

Any achievable point does not exceed the point-to-point channel capacity of each user. So, a sufficient requirement is to bound the distance of the point that simultaneously achieves all point-to-point channel capacities, i.e the total information rate transmitted to the $t$-th receiver is
\[ C_{AWGN}(P / N_t) \triangleq  \log (1 + P / N_t)\]
To achieve this simultaneously, one should carry information to each receiver at a rate equals to the difference between its point-to-point capacity and the point-to-point capacity of the next more-degraded receiver, namely
\[ C_{AWGN}(P / N_t) - C_{AWGN}(P / N_{t+1}). \]
Using a power allocation according to $\{\alpha_t\}_0^T$, i.e, a ratio of $\alpha_t - \alpha_{t-1}$ of the total power $P$ is used for the $t$-th message, the classical AWGN broadcast channel separation-based capacity region is
\[ \log \left(1 +  \frac{(\alpha_t - \alpha_{t-1})P}{N_t + \alpha_{t-1}P} \right) = C_{AWGN}\left(\frac{(\alpha_t - \alpha_{t-1})P}{N_t + \alpha_{t-1}P} \right).\]

As in Theorem \ref{theorem:refined_bound} we use separation to inductively attain the requirement of the receivers, but this time starting from the first receiver. For each receiver $t$ except for the last $T$-th one, we will allocate the exact power required to carry the message, namely we choose $\alpha_t$ satisfying
\begin{multline*}
C_{AWGN}\left(\frac{(\alpha_t - \alpha_{t-1})P}{N_t + \alpha_{t-1}P} \right) =
\\ C_{AWGN}(P / N_t) - C_{AWGN}(P / N_{t+1}).
\end{multline*}
This yields the recursive relation
\[ \alpha_t(1+P/N_{t+1}) = \alpha_{t-1}(1+P/N_t) + 1 - N_t/N_{t+1} \]
and its solution for $\alpha_0 = 0$ is
\[ \alpha_{T-1}(1+P/N_T) = \sum_{t = 1}^{T-1} 1 - N_t/N_{t+1}. \]
The last, most degraded receiver uses $\alpha_T = 1$ achieving an information rate per channel use of
\[ C_{AWGN}\left(\frac{(1 - \alpha_{T-1})P}{N_T + \alpha_{T-1}P} \right) \]
Subtracting this rate from the rate of the best achievable point, yields the following information gap
\begin{multline*}
C_{AWGN}(P / N_T) - C_{AWGN}\left(\frac{(1 - \alpha_{T-1})P}{N_T + \alpha_{T-1}P} \right) =
\\ C_{AWGN}(\alpha_{T-1} P / N_T)
\end{multline*}
\end{proof}

\begin{cor}
The classical separation-based capacity region of the AWGN broadcast channel with $T$ receivers cannot be exceeded with more than $ \log_2 T $ bits per complex degree of freedom.
\end{cor}
\begin{proof}
This result immediately follows from Theorem \ref{theorem:awgn} by using
\[ \frac{P}{P + N_T} < 1 \]
\[ \forall t=1,...,T-1: \;\; 1 - \frac{N_t}{N_{t+1}} < 1 \]
\end{proof}

\begin{cor}
The classical separation-based capacity region of the AWGN broadcast channel with any number of receivers and a bounded noise range $[N_{min}, N_{max}]$ cannot be exceeded with more than $ \log_2 (1+\ln(N_{max}/N_{min})) $ bits per complex degree of freedom.
\end{cor}
\begin{proof}
The corrolary follows immediately from Theorem \ref{theorem:awgn} by using the asymptotic derivation used in the proof of Theorem \ref{theorem:refined_bound}.
\end{proof}

We note that the bound of Theorem \ref{theorem:awgn} can be improved by refining the trade-off between the last receiver and all previous receivers. For example, one can use linear time-sharing to spread the rate gap of the most remote receiver to its predecessors.

\ifdefined\whole

\begin{figure*}
\includegraphics[width=\linewidth]{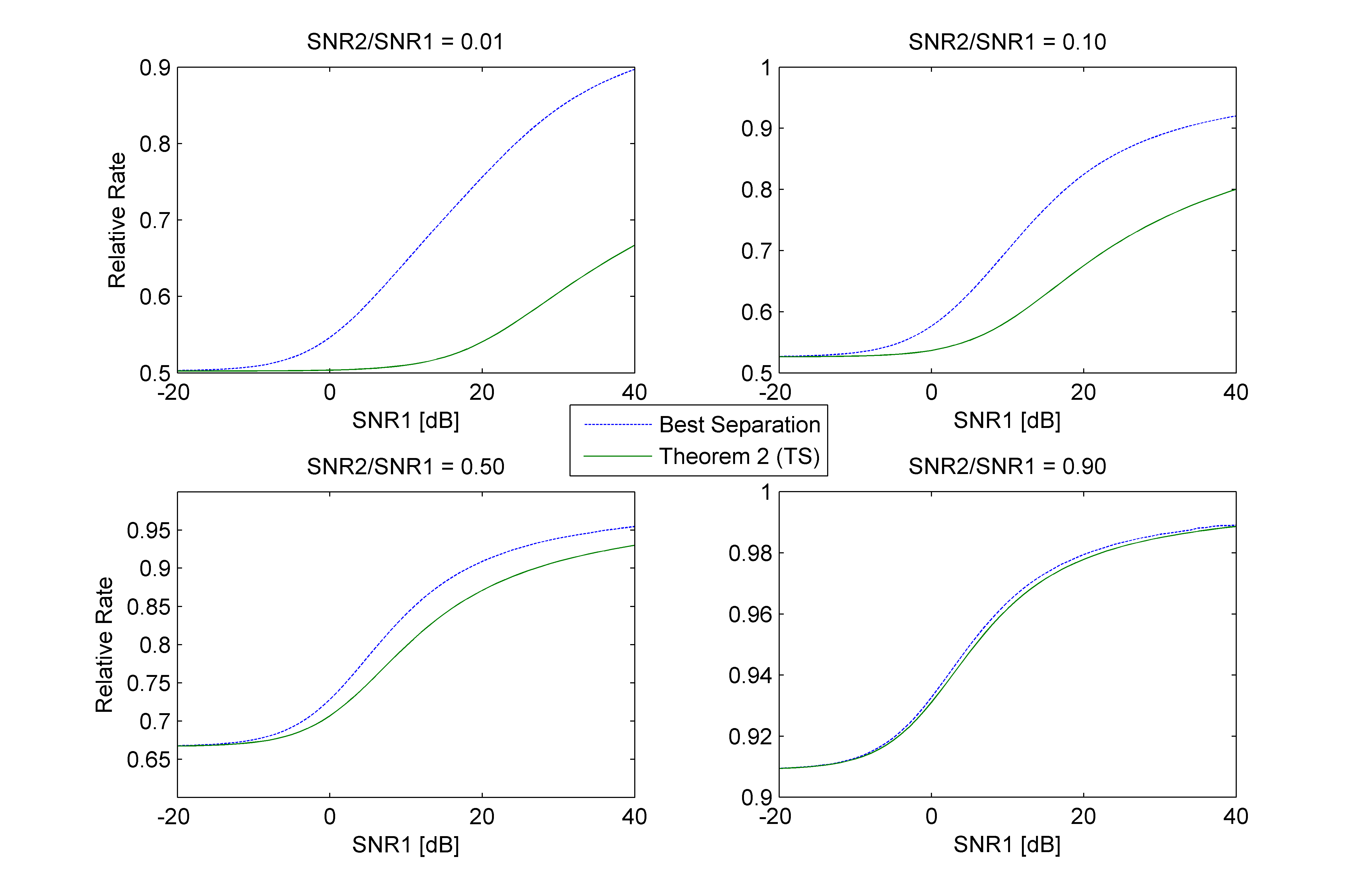}
\caption{Rate relative to joint-coding for two users.}
\label{fig:awgn}
\end{figure*}

\begin{figure*}
\includegraphics[width=\linewidth]{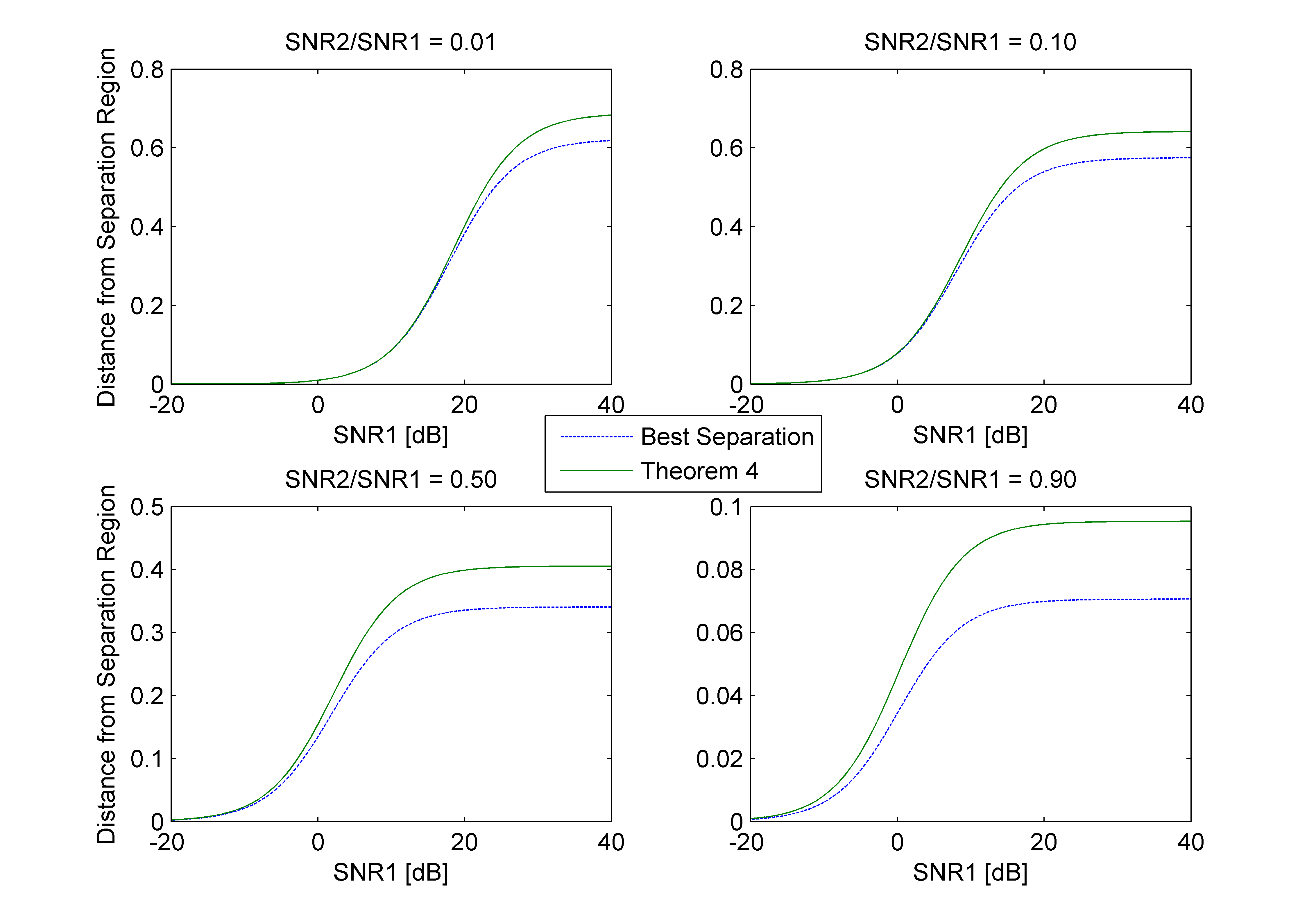}
\caption{Distance of the best achievable point from the separation-based capacity region for two users.}
\label{fig:awgn_distance}
\end{figure*}

\subsection{Examples}
\label{subsec:awgn_finite}

We will use the best achievable point as a reference for analyzing the performance of the optimal separation method and the results of the bounds given in Theorems \ref{theorem:naive_bound}, \ref{theorem:refined_bound} and \ref{theorem:awgn}. Fig. \ref{fig:awgn} plots the achievable rates of both an optimal separation-based scheme as well as the result of Theorem \ref{theorem:refined_bound} achieved by time-sharing, for the case of two receivers ($T=2$), relative to the best achievable point. The graphs are plotted for four values of noise ratios $N_1/N_2 = 0.01, 0.1, 0.5, 0.9$ and x-axis is the signal-to-noise ratio (SNR) of the better receiver $P/N_1$.  Fig. \ref{fig:awgn_distance} plots the distance of the best achievable point from the separation-base capacity region and the bound on this distance resulting from Theorem \ref{theorem:awgn}.

Fig. \ref{fig:awgn} illustrates that as the SNR decreases the performance of time-sharing and optimal separation are comparable, thus in low-SNR time-sharing becomes a good separation-based method. Nevertheless, we see that using separation will have a large negative impact on the performance in low SNR values, up to a factor of $2$, as suggested by Theorem \ref{theorem:naive_bound}. But according to both Fig. \ref{fig:awgn} and Fig. \ref{fig:awgn_distance} as the SNR increases separation method approaches the joint coding performance, because although the rates grow their difference remains bounded.

\begin{figure}
\includegraphics[width=\linewidth]{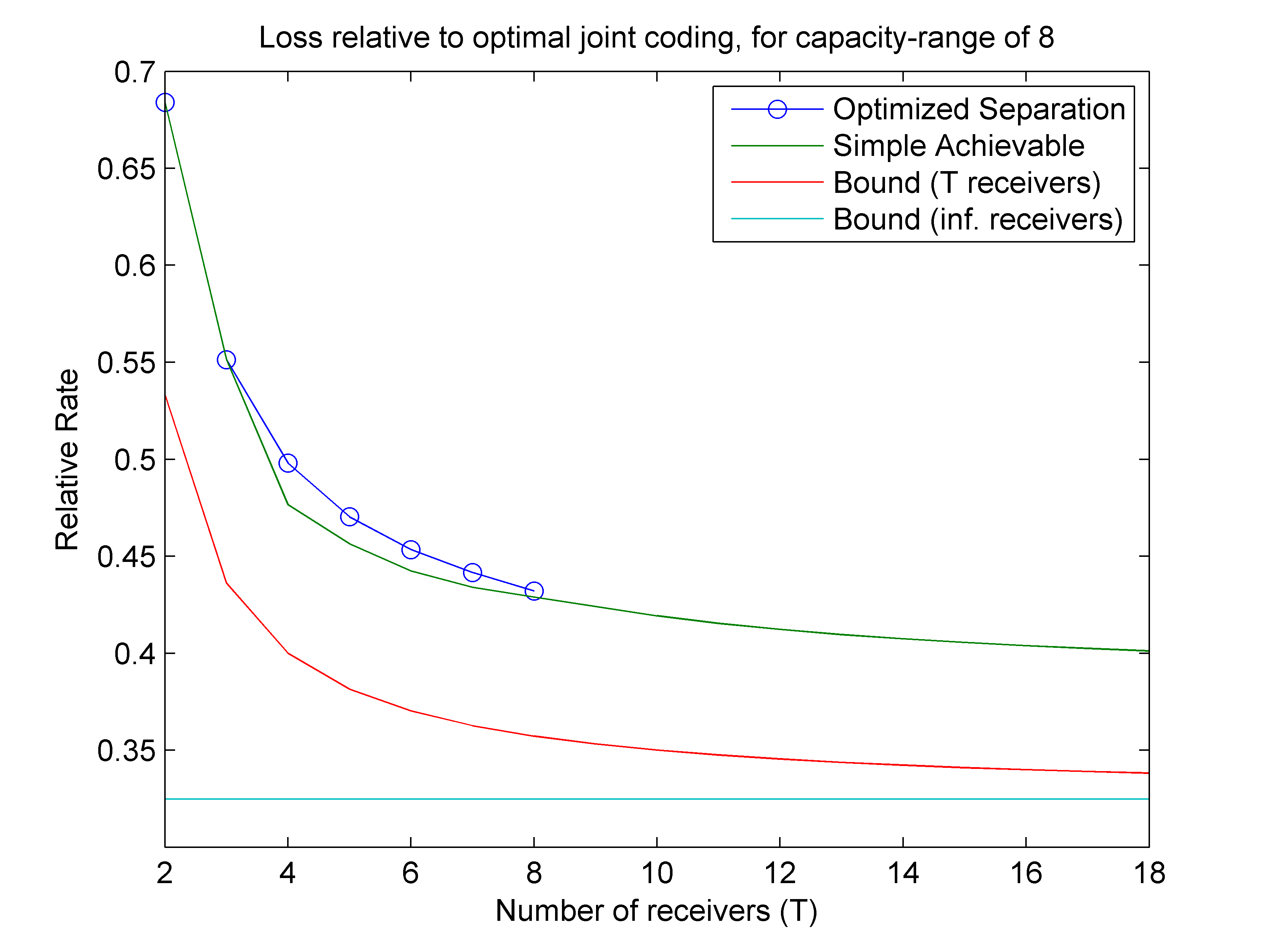}
\caption{Rate relative to joint-coding for many users.}
\label{fig:users}
\end{figure}

Fig. \ref{fig:users} illustrates the behavior as the number of receivers increases. Vertical axis values are rates relative to the rate of optimal joint source-channel coding. We can see the optimal separation-based method up to eight users. Beyond this the optimal strategy is too hard to compute, so we use a simpler near-optimal strategy \footnote{Instead of optimizing $T$ independent power values allocated for each receiver, the near optimal strategy uses optimization with only two degrees of freedom. It assumes that power amounts are a linear function of $t$, except for the power allocated to transmit the message to the most remote receiver, these three optimization variables combined with the total power constraint result in a traceable optimization with two degrees of freedom.}. The results are given for choosing a geometric sequence of capacities as suggested by Theorem \ref{theorem:refined_bound}, and for $C_{max}/C_{min} = R_{max}/R_{min} = 8$ and values are chosen such that their geometric mean is $1$ bits per symbol. We can see that for large number of receivers the proposed time sharing bound achieves a third of the joint coding performance, while the best separation-based strategy achieves about $5$ percent more.

\else

We refer the reader to the complete version of this paper \cite{Complete} for further discussion and examples of the AWGN scenario.

\fi

\section{Conclusions and Future Work}
\label{sec:future}

This work proposes useful bounds on the rate lost due to source-channel separation and time-sharing for refining a source through a broadcast channel. Following our results we argue that separation-based codes as well as applying time-sharing are both a practical compromise between design and computational complexity and transmission performance, regardless of the amount of receivers.

All outcomes for broadcast channels are naturally applicable for the scenario of point-to-point source transmission over a channel with channel quality unknown in the transmitter, when the receiver should recover the best possible source quality for the actual quality of the channel.

Section \ref{sec:bounds} provides a set of general bounds. Theorem \ref{theorem:naive_bound} can be trivially applied to any source broadcasting scenario, Theorem \ref{theorem:refined_bound} requires only the values of the classical point-to-point capacities or rate distortion, and Theorem \ref{theorem:capacity_bound} combines both into a tighter bound.
Section \ref{sec:awgn} discusses the case of AWGN channels, providing an \emph{additive} bound on the capacity region. This result is complementary to the bounds suggested by \cite{Tian09}, studying AWGN channels instead of Gaussian sources. Note that while Theorem \ref{theorem:awgn} bounds the rate loss for a single channel use, \cite{Tian09} and \cite{Tian10} provide bounds in terms of loss per source symbol.

Our work suggests some future research. The approach of \cite{Tian10} can be used to construct additive bounds for AWGN-like channels, e.g AWGN channel with general input constraint.

The use of time-sharing for constructing lower-bounds of separation-loss may be applicable in other joint source-channel problems where separation is not optimal, such as a MAC with correlated sources and solving a CEO problem through a MAC. The approach may be used to compute some general single-letter lower-bounds for other problems with no known single-letter capacity region, such as the interference channel.

\ifdefined\whole
Another insight from our work, and also from the work of Etkin et al. \cite{within} is the special property of AWGN multi-terminal networks. We make the following conjecture:

\begin{conjecture}
\label{conj:awgn}
For any multi-terminal AWGN information network, an achievable single letter solution is within a bounded distance from the optimal solution, regardless of the SNR value.
\end{conjecture}

For example, in the $K$-user interference channel the interference alignment method proposed by Cadambe and Jafar \cite{alignment} achieves an optimal asymptotic performance for high-SNR AWGN interference channels, but the absolute gap of this method from the optimal performance is not proven to be bounded. Our conjecture may lead to a single letter generalized Han-Kobayashi scheme in the spirit of interference alignment achieving the capacity region to within constant gap.

Proving Conjecture \ref{conj:awgn} may seem difficult, but we suggest the use of deterministic channels to establish a proof. For many information network problems the single letter solution is optimal when all channels are deterministic. Thus, we suggest that a reduction can be made from any high-SNR AWGN problem to a problem with deterministic channels. If the analogous deterministic problem has a single letter optimal solution, then Conjecture \ref{conj:awgn} can be proven by bounding the difference between the original AWGN problem and the deterministic analogy. This kind of argument was already established for interference channels with $2$ users \cite{deterministic_awgn} and for channels with partial interference \cite{many2one}.

\fi

\end{document}